\documentclass{llncs}

\usepackage{amsmath, amssymb}
\usepackage{tikz}
\usetikzlibrary{arrows,automata}
\usepackage[pagewise]{lineno}
\usepackage{subfig}
\usepackage{algorithmicx}
\usepackage{algpseudocode}

\begin{document}

\title{On nonpermutational transformation semigroups with an application to syntactic complexity}
\author{Szabolcs~Iv\'an \and Judit~Nagy-Gy\"orgy}
\institute{University of Szeged}

\setlength{\parindent}{0pt}
\setlength{\parskip}{6pt}
\linenumbers

\let\doendproof\endproof
\renewcommand\endproof{~\hfill$\qed$\doendproof}

\maketitle

\begin{abstract}
We give an upper bound of $n((n-1)!-(n-3)!)$ for the possible largest size of a subsemigroup of the
full transformational semigroup over $n$ elements
consisting only of nonpermutational transformations.
As an application we gain the same upper bound for the syntactic complexity of
(generalized) definite languages as well.
\end{abstract}

\section{Introduction}

A language is generalized definite if membership can be decided for a word by looking at its prefix and suffix of a given constant length.
Generalized definite languages and automata were introduced by Ginzburg~\cite{ginzburg} in 1966
and further studied in e.g.~\cite{ciricimrehsteinby,gecsegimreh,petkovic,steinby}.
This language class is strictly contained within the class of star-free languages, lying on the first level of the
dot-depth hierarchy~\cite{dotdepth}.
This class possess a characterization in terms of its syntactic semigroup \cite{perrin}: a regular language is generalized definite
if and only if its syntactic semigroup is locally trivial if and only if it satisfies a certain identity $x^\omega yx^\omega=x^\omega$.
This characterization is hardly efficient by itself when the language is given by its minimal automaton,
since the syntactic semigroup can be much larger than the automaton (a construction for a definite language with state complexity
-- that is, the number of states of its minimal automaton -- $n$
and syntactic complexity -- that is, the size of the transition semigroup of its minimal automaton --
$\lfloor e(n-1)!\rfloor$ is explicit in \cite{brzozo}).
However, as stated in~\cite{handbook}, Sec.~5.4, it is usually not necessary to compute the (ordered) syntactic semigroup
but most of the time one can develop a more efficient algorithm by analyzing the minimal automaton.
As an example for this line of research, recently, the authors of \cite{klima-polak} gave a nice characterization of minimal automata of
piecewise testable languages, yielding a quadratic-time decision algorithm, matching an alternative (but of course equivalent)
earlier (also quadratic) characterization of~\cite{trahtman} which improved the $\mathcal{O}(n^5)$ bound of~\cite{Stern}.

There is an ongoing line of research for syntactic complexity of regular languages.
In general, a regular language with state complexity $n$ can have a syntactic complexity of $n^n$, already in the case
when there are only three input letters. There are at least two possible modifications of the problem:
one option is to consider the case when the input alphabet is binary (e.g. as done in \cite{holzer-konig,krawetz-lawrence-shallit}).
The second option is to study a strict subclass of regular languages.
In this case, the syntactic complexity of a class $\mathcal{C}$ of languages is a function $n\mapsto f(n)$,
with $f(n)$ being the maximal syntactic complexity a member of $\mathcal{C}$ can have whose state complexity
is (at most) $n$. The syntactic complexity of several language classes, e.g. (co)finite, reverse definite,
bifix--, factor-- and subword-free languages etc. is precisely determined in \cite{limsc}.
However, the exact syntactic complexity of the (generalized) definite languages and that of the star-free languages
(as well as the locally testable or the locally threshold testable languages) is not known yet.

In this note we give an upper bound for the maximal size of a subsemigroup of $T_n$, the 
transformation semigroup of $\{1,\ldots,n\}$, consisting of ``nonpermutational'' transformations only.
These are exactly the (transformation) semigroups satisfying the identity $yx^\omega=x^\omega$. 
It is known that a language is definite iff its syntactic semigroup satisfies the same identity;
thus as a corollary we get that the same bound is also an upper bound for the syntactic complexity
of definite languages. 

We also give a forbidden pattern characterization for the generalized definite languages in terms of the minimal automaton,
and analyze the complexity of the decision problem whether a given automaton recognizes a generalized definite language,
yielding an $\mathbf{NL}$-completeness result (with respect to logspace reductions) as well as a deterministic decision
procedure running in $\mathcal{O}(n^2)$ time (on a RAM machine).
Analyzing the structure of their minimal automata
we conclude that the syntactic complexity of generalized definite languages coincide with
that of definite languages.

\section{Notation}
When $n\geq 0$ is an integer, $[n]$ stands for the set $\{1,\ldots,n\}$.
For the sets $A$ and $B$, $A^B$ denotes the set of all functions $f:B\to A$. When $f\in A^B$ and $C\subseteq B$,
then $f|_C\in A^C$ denotes the restriction of $f$ to $C$. When $A_1,\ldots,A_n$ are disjoint sets, $A$ is a set
and for each $i\in[n]$, $f_i:A_i\to A$ is a function, then the \emph{source tupling} of $f_1,\ldots,f_n$ is
the function $[f_1,\ldots,f_n]:\bigl(\mathop\bigcup\limits_{i\in[n]}A_i\bigr)\to A$ with
$a[f_1,\ldots,f_n]=af_i$ for the
unique $i$ with $a\in A_i$.\marginpar{$[f_1,\ldots,f_n]$: source tupling}

$T_n$ is the transformation semigroup of $[n]$ (i.e. $[n]^{[n]}$), where composition
is understood as $p(fg):=(pf)g$ for $p\in[n]$ and $f,g:[n]\to[n]$ (i.e., transformations of $[n]$
act on $[n]$ from the right to ease notation in the automata-related part of the paper).
Elements of $T_n$ are often written as $n$-ary vectors as usual, e.g. $f=(1,3,3,2)$ is
the member of $T_4$ with $1f=1$, $2f=3$, $3f=3$ and $4f=2$.

When $f:A\to A$ is a transformation of a set $A$, and $X$ is a subset of $A$, then
$Xf$ denotes the subset $\{xf:x\in X\}$ of $A$.

A transformation $f:A\to A$ of a (finite) set $A$ is \marginpar{nonpermutational function}{\emph{nonpermutational}} if
$Xf=X$ implies $|X|=1$ for any nonempty $X\subseteq A$. Otherwise it's \emph{permutational}.
\marginpar{$NP_n$}{$NP_n$} stands for the set of all nonpermutational transformations of $[n]$.

Another class of functions used in the paper is that of the \emph{elevating} functions: for the integers $0< k\leq n$, a function
$f:[k]\to[n]$ is elevating if $i\leq if$ for each $i\in[k]$
with equality allowed only in the case when $i=n$ (note that this also implies $k=n$ as well).\marginpar{elevating function}

We assume the reader is familiar with the standard notions of automata and language theory,
but still we give a summary for the notation.

An \emph{alphabet} is a nonempty finite set $\Sigma$.
The set of \emph{words} over $\Sigma$ is denoted $\Sigma^*$, while $\Sigma^+$ stands for the set of \emph{nonempty words}.
The \emph{empty word} is denoted $\varepsilon$.
A \emph{language} over $\Sigma$ is an arbitrary set $L\subseteq\Sigma^*$ of $\Sigma$-words.

A (finite) \emph{automaton} (over $\Sigma$) is a system $\mathbb{A}=(Q,\Sigma,\delta,q_0,F)$ where
$Q$ is the finite set of states,
$q_0\in Q$ is the start state,
$F\subseteq Q$ is the set of final (or accepting) states,
and $\delta:Q\times \Sigma\to Q$ is the transition function.
The transition function $\delta$ extends in a unique way to a right action of the monoid $\Sigma^*$ on $Q$,
also denoted $\delta$ for ease of notation%
.
When $\delta$ is understood, we write $q\cdot u$, or simply $qu$ for $\delta(q,u)$.
Moreover, when $C\subseteq Q$ is a subset of states and $u\in\Sigma^*$ is a word, let $Cu$ stand for the set $\{pu:p\in C\}$
and when $L$ is a language, $CL=\{pu:p\in C,u\in L\}$.
The \emph{language recognized by $\mathbb{A}$} is $L(\mathbb{A})=\{x\in\Sigma^*:q_0x\in F\}$.
A language is \emph{regular} if it can be recognized by some finite automaton.

The state $q\in Q$ is \emph{reachable} from a state $p\in Q$ in $\mathbb{A}$, denoted $p\preceq_{\mathbb{A}} q$,
or just $p\preceq q$ if there is no danger of confusion, if $pu=q$ for some $u\in\Sigma^*$.
An automaton is \emph{connected} if its states are all reachable from its start state.

Two states $p$ and $q$ of $\mathbb{A}$ are \emph{distinguishable} if there exists a word $u\in\Sigma^*$ such that
exactly one of $pu$ and $qu$ belongs to $F$. In this case we say that $u$ \emph{separates} $p$ and $q$.
A connected automaton is called \emph{reduced} if each pair of distinct states is distinguishable.

It is known that for each regular language $L$ there exists a reduced automaton, unique up to isomorphism, recognizing $L$.
This automaton $\mathbb{A}_L$ can be computed from any automaton recognizing $L$ by an efficient algorithm called minimization and is called the
\emph{minimal automaton} of $L$.\marginpar{$\mathbb{A}_L$}

The classes of the equivalence relation $p\sim q\ \Leftrightarrow p\preceq q\textrm{ and }q\preceq p$ are called \emph{components} of $\mathbb{A}$.
A component $C$ is \emph{trivial} if $C=\{p\}$ for some state $p$ such that $pa\neq p$ for any $a\in\Sigma$,
and is a \emph{sink} if $C\Sigma\subseteq C$. It is clear that each automaton has at least one sink and sinks are never trivial.\marginpar{(trivial) components and sinks}
The \emph{component graph} $\Gamma(\mathbb{A})$ of $\mathbb{A}$ is an edge-labelled directed graph $(V,E,\ell)$ along with a mapping $c:Q\to V$
where $V$ is the set of the $\sim$-classes of $\mathbb{A}$, the mapping $c$ associates to each state $q$ its class $q/\sim=\{p:p\sim q\}$
and for two classes $p/\sim$ and $q/\sim$ there exists an edge from $p/\sim$ to $q/\sim$ labelled by $a\in\Sigma$
if and only if $p'a=q'$ for some $p'\sim p$, $q'\sim q$.
It is known that the component graph can be constructed from $\mathbb{A}$ in linear time.
Note that the mapping $c$ is redundant but it gives a possibility for determining whether $p\sim q$ holds in constant time on a
RAM machine, provided $Q=[n]$ for some $n>0$ and $c$ is stored as an array.

When $\mathbb{A}=(Q,\Sigma,\delta,q_0,F)$ is an automaton, its \emph{transformation semigroup} $\mathcal{T}(\mathbb{A})$
consists of the set of transformations of $Q$ induced by nonempty words, i.e.
$\mathcal{T}(\mathbb{A})=\{u^\mathbb{A}:u\in\Sigma^+\}$ where $u^\mathbb{A}:Q\to Q$ is the transformation defined as $q\mapsto qu$.
The \emph{state complexity} $\textrm{stc}(L)$ of a regular language $L$ is the number of states of its minimal automaton $\mathbb{A}_L$ while its \emph{syntactic complexity} $\textrm{syc}(L)$ is the
cardinality of its transformation semigroup $\mathcal{T}(\mathbb{A}_L)$.
The \emph{syntactic complexity} of a \emph{class} of languages
$C$ is a function $f:\mathbb{N}\to\mathbb{N}$ defined as
\[f(n)=\max\{\mathrm{syc}(L):L\in C,\mathrm{stc}(L)\leq n\},\]
i.e. $f(n)$ is the maximal size that the transformation semigroup of a minimal automaton of a language
belonging to $C$ can have, provided the automaton has at most $n$ states.

\section{Semigroups of nonpermutational transformations}
Observe that $NP_n$ is not a semigroup (i.e., not closed under composition) when $n>2$.
Indeed, if $f=(2,3,3)$ and $g=(1,1,2)$ (both being nonpermutational), then
their product $fg=(1,2,2)$ is permutational with $\{1,2\}fg=\{1,2\}$.
(See Figure~\ref{fig-fgperm}.)

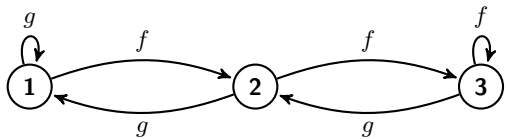
\begin{figure}[h!]\centering
\begin{tikzpicture}[->,>=stealth',shorten >=1pt,auto,node distance=3cm,
  thick,main node/.style={circle,draw,font=\sffamily\bfseries}]

  \node[main node] (1) {1};
  \node[main node] (2) [right of=1] {2};
  \node[main node] (3) [right of=2] {3};

  \path[->] (1) edge [bend left=20] node {$f$} (2);
  \path[->] (2) edge [bend left=20] node {$f$} (3);
  \path[->] (2) edge [bend left=20] node {$g$} (1);
  \path[->] (3) edge [bend left=20] node {$g$} (2);
  \path[->] (1) edge [loop above] node {$g$} (1);
  \path[->] (3) edge [loop above] node {$f$} (3);
  
  \end{tikzpicture}
  \caption{$f$ and $g$ are nonpermutational, $fg$ is permutational}
  \label{fig-fgperm}
\end{figure}

Thus, the following question is nontrivial: how large a subsemigroup of $T_n$,
which consists only of nonpermutational transformations can be?
The obvious upper bound is $n^n$, the size of $T_n$.

{\bf As a first step we give an upper bound of $n^{n-2}$.}
Observe that the following are equivalent for a function $f:[n]\to [n]$:
\begin{enumerate}
\item[i)] $f$ is nonpermutational;
\item[ii)] the graph of $f$ is a rooted tree with edges directed towards the root,
  and with a loop edge attached on the root;
\item[iii)] $f^\omega$, the unique idempotent power of $f$ is a constant function.
\end{enumerate}
Here ``the graph of $f$'' is of course the directed graph $\Gamma_f$
on vertex set $[n]$ and with $(i,j)$ being an edge iff $if=j$.

Indeed, assume $f$ is nonpermutational. Let $X$ be the set of all nodes of $\Gamma_f$
lying on some closed path. (Since each node of the finite graph $\Gamma_f$ has outdegree $1$,
$X$ is nonempty.) Then $Xf=X$, thus $|X|=1$, i.e.
$f$ has a unique fixed point \marginpar{$\mathrm{Fix}(f)$}$\mathrm{Fix}(f)$ and apart from the loop edge on $\mathrm{Fix}(f)$,
$\Gamma_f$ is a directed acyclic graph (DAG) with each node distinct from $\mathrm{Fix}(f)$
having outdegree $1$ -- that is, a tree rooted at $\mathrm{Fix}(f)$, with edges directed towards
the root, showing i) $\to$ ii). Then $f^n$ is a constant function with value $\mathrm{Fix}(f)$,
showing ii) $\to$ iii); finally, if $Xf=X$ for some nonempty $X\subseteq[n]$,
then $Xf^\omega=X$, showing $|X|=1$ since the image of $f^\omega$ is a singleton.

Now from ii) we get that the members of $NP_n$ are exactly the rooted trees with edges directed
towards the root on which a loop edge is attached -- we call such a graph an inverted looped arborescence\footnote{For comparison, an arborescence is a rooted tree with its
  edges being directed \emph{away from} the root. Adding a loop edge to the root yields
  a ``looped arborescence''. However, we were unable to find a name in the literature for
  transposes of arborescences -- if there exists some, it would be better to use that name
  instead of ``nonpermutational''. Coining the term ``ecnecserobra'' is out of question. ``Ultimately constant'' would be also an option. We would be thankful for the referees to point out an existing term in the literature.},
or ILA for short. By Cayley's theorem on the number of labeled rooted trees over $n$ nodes,
the number of all ILAs (i.e., $|NP_n|$) is $n^{n-2}$, giving a slightly better upper bound.

{\bf To achieve an upper bound of $n!$}, suppose $S\subseteq NP_n$ is a subsemigroup of $T_n$.
For $i\in[n]$, let $S_i\subseteq S$ be the subsemigroup $\{f\in S:\mathrm{Fix}(f)=i\}$ of $S$.
Note that $S_i$ is indeed a semigroup: by assumption, $S$ is closed under composition and
consists of nonpermutational transformations only, moreover, if $i$ is the common (unique)
fixed point of $f$ and $g$, then it is also a fixed point of $fg$ as well, thus $S_i$
is closed under composition.

{\bf We give an upper bound of $(n-1)!$ for $|S_i|$}, $i\in[n]$, yielding $|S|\leq n!$.
To this end, let $\Gamma_i$ be the graph on vertex set $[n]$ with $(j,k)$ being an edge
iff $jf=k$ for some $f\in S_i$. Then, apart for the trivial case when $S_i=\emptyset$,
$(i,i)$ is an edge in $\Gamma_i$, moreover $i$ is
a sink (since $if=i$ for each $f\in S_i$). Note that in the case when $S_i=\emptyset$,
$|S_i|=0\leq(n-1)!$ clearly holds. Observe that $\Gamma_i$ is transitive, since
if $(j,k)$ and $(k,\ell)$ are edges of $\Gamma_i$, then $jf=k$ and $kg=\ell$ for some
$f,g\in S_i$; since $S_i$ is a semigroup, $fg$ is also in $S_i$ thus $(j,\ell)$ is
also an edge in $\Gamma_i$. Now assume some node $j\in[n]$ is in a nontrivial strongly connected component
(SCC) of $\Gamma_i$, i.e. $j$ lies on some closed path. By transitivity, $(j,j)$ is
an edge of $\Gamma_i$, thus $jf=j$ for some $f\in S_i$, thus $j=i$ since $i=\mathrm{Fix}(f)$
is the unique fixed point of $f\in S_i$. Hence by dropping the edge $(i,i)$ we get
a DAG again, thus $\Gamma_i$ (viewed as a relation) is a strict partial ordering of $[n]$
with largest element $i$. Let $\prec_i$ stand for this partial ordering, i.e.,
let $j\prec_i k$ if and only if $j\neq i$ and $jf=k$ for some $f\in S_i$.
Let us also fix some arbitrary total ordering $<_i$ extending $\prec_i$ and write
the members of $[n]$ in the order $a_{i,1}<_ia_{i,2}<_i\ldots<_ia_{i,n}=i$. Then for any
$f\in S_i$ and $1\leq j<n$ we have $a_{i,j}<_ia_{i,j}f$, and $a_{i,n}f=a_{i,n}$.
Since the number of functions $f:[n]\to [n]$ satisfying this constraint is $(n-1)!$
($a_{i,1}$ can get $(n-1)$ different possible values, $a_{i,2}$ can get $(n-2)$ etc.),
we immediately get $|S_i|\leq (n-1)!$ as well, yielding $|S|\leq n!$.

{\bf Via a somewhat cumbersome case analysis we can sharpen this upper bound to $n((n-1)!-(n-3)!)$.}
Without loss of generality assume that $S_n$ is (one of) the largest of the semigroups $S_i$
and that $<_n$ is the usual ordering $<$ of $[n]$ (we can achieve this by a suitable bijection).

\begin{lemma}
Suppose for each $i<j$ and $k<\ell$ with $i\neq k$
there exists a function $f\in S_n$ with $if=j$ and $kf=\ell$.

Then the following holds for each $i,j\in[n]$ and $f\in S_i$:
\begin{enumerate}
\item[i)] if $j<i$, then $j<jf$;
\item[ii)] if $i\leq j$, then $jf=i$.
\end{enumerate}
\label{lem-main}
\end{lemma}

\begin{proof}
By assumption, the statements clearly hold for $i=n$. Let $i<n$ be arbitrary and $f\in S_i$
a transformation. Clearly $if=i$ by the definition of $S_i$.
Also, $nf<n$ since $i\neq n$ is the unique fixed point of $f$.

Suppose $jf<j$ for some $j$. Then $jf=nf$ has to hold: if $jf\neq nf$, then by assumption
$jfg=j$ and $nfg=n$ for some $g\in S_n$, thus both $j$ and $n$ are distinct fixed points
of $fg$, a contradiction. (See Figure~\ref{fig-points-one}.)
This implies in particular that $j\leq jf$ for each $j<nf$.

Also, if $nf<i$, then $nfg=i$ and $ig=n$ for some $g\in S_n$, in which case $fgfg$ has
two distinct fixed points $n$ and $i$, a contradiction. (See Figure~\ref{fig-points-one}.)
Thus $i\leq nf$.

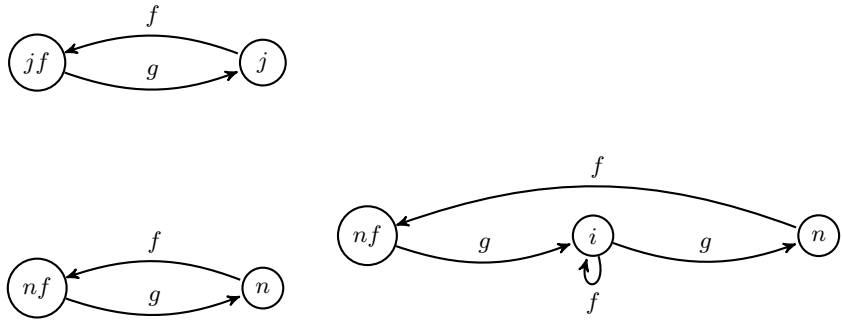
\begin{figure}[h!]\centering
\begin{tikzpicture}[->,>=stealth',shorten >=1pt,auto,node distance=3cm,
  thick,main node/.style={circle,draw,font=\sffamily\bfseries}]
  \node[main node] (1) {$jf$};
  \node[main node] (2) [right of=1] {$j$};
  \node[main node] (3) [below of=1] {$nf$};
  \node[main node] (4) [right of=3] {$n$};
  \path[->] (1) edge [bend right=20] node {$g$} (2);
  \path[<-] (1) edge [bend left=20] node {$f$} (2);
  \path[->] (3) edge [bend right=20] node {$g$} (4);
  \path[<-] (3) edge [bend left=20] node {$f$} (4);
  \end{tikzpicture}
\hfil
\begin{tikzpicture}[->,>=stealth',shorten >=1pt,auto,node distance=3cm,
  thick,main node/.style={circle,draw,font=\sffamily\bfseries}]
  \node[main node] (1) {$nf$};
  \node[main node] (2) [right of=1] {$i$};
  \node[main node] (3) [right of=2] {$n$};
  \path[->] (1) edge [bend right=20] node {$g$} (2);
  \path[->] (2) edge [bend right=20] node {$g$} (3);
  \path[->] (2) edge [loop below] node {$f$} (2);
  \path[<-] (1) edge [bend left=20] node {$f$} (3);
  \end{tikzpicture}
  \caption{Left: if $jf<j$, $jf\neq nf$, then $fg$ has two fixed points. Right: If $nf<i$, then $fgfg$ has two fixed points}
  \label{fig-points-one}
\end{figure}
Assume $i<nf$. Then (since $nf^n=i<nf$) there is some $k>0$ such that $nf^{k+1}<nf$.
If $k$ is chosen to be the smallest possible such $k$, then $nf\leq nf^k$, yielding
$(nf^k)f<nf\leq nf^k$, a contradiction (by $(nf^k)f<nf^k$, it should hold that
$(nf^k)f=nf$, see Figure~\ref{fig-points-two}). Hence $i=nf$ is the unique fixed point of $f$ and for each $j<i$,
$j<jf$ indeed has to hold, showing i).
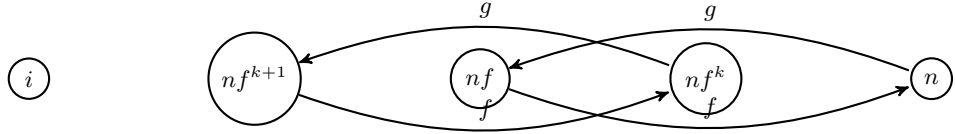
\begin{figure}[h!]
\begin{tikzpicture}[->,>=stealth',shorten >=1pt,auto,node distance=3cm,
  thick,main node/.style={circle,draw,font=\sffamily\bfseries}]
  \node[main node] (1) {$i$};
  \node[main node] (2) [right of=1] {$nf^{k+1}$};
  \node[main node] (3) [right of=2] {$nf$};
  \node[main node] (4) [right of=3] {$nf^{k}$};
  \node[main node] (5) [right of=4] {$n$};
  
  \path[->] (2) edge [bend right=20] node {$f$} (4);
  \path[<-] (2) edge [bend left=20] node {$g$} (4);
  \path[->] (3) edge [bend right=20] node {$f$} (5);
  \path[<-] (3) edge [bend left=20] node {$g$} (5);
  \end{tikzpicture}
  \caption{If $i<nf$, then $fg$ has two distinct fixed points}
  \label{fig-points-two}
\end{figure}

Finally, assume $i<j<jf$. Then $ig=j$ and $jfg=n$ for some $g\in S_n$ (if $jf=n$, then this
latter case always gets satisfied, otherwise it's by assumption on $S_n$), and
$fgfg$ has two distinct fixed points $j$ and $n$.
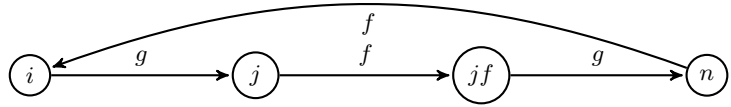
\begin{figure}[h!]
\centering
\begin{tikzpicture}[->,>=stealth',shorten >=1pt,auto,node distance=3cm,
  thick,main node/.style={circle,draw,font=\sffamily\bfseries}]
  \node[main node] (1) {$i$};
  \node[main node] (2) [right of=1] {$j$};
  \node[main node] (3) [right of=2] {$jf$};
  \node[main node] (4) [right of=3] {$n$};
  
  \path[->] (1) edge node {$g$} (2);
  \path[->] (2) edge node {$f$} (3);
  \path[->] (3) edge node {$g$} (4);
  \path[->] (4) edge [bend right=20] node {$f$} (1);
  \end{tikzpicture}
  \caption{If $i<j<jf$, then $fgfg$ has two distinct fixed points}
\end{figure}
Thus we have indeed shown that $nf=i$ is the unique fixed point of $f$, $j<jf$ for each $i<j$
and $jf=i$ for each $i\leq j\leq n$.
\end{proof}
Lemma~\ref{lem-main} has the following corollary:
\begin{theorem}
\label{thm-main}
The cardinality of any subsemigroup $S$ of $T_n$ consisting only of nonpermutational
transformations is at most $n((n-1)!-(n-3)!)$.
\end{theorem}
\begin{proof}
As before, let $S_i$ stand for $\{f\in S:\mathrm{Fix}(f)=i\}$ and without loss of
generality we assume that amongst them $S_n$ is one of the largest one, moreover
$<_n$ coincides with $<$.

If for each $i<j$ and $i'<j'$ with $i\neq i'$ there is some $f\in S_n$ with
$if=j$ and $i'f=j'$, then by Lemma~\ref{lem-main} $S_i$ can consist of at most
$(n-1)(n-2)\ldots(n-i-1)=\frac{(n-1)!}{(n-i)!}$ elements (we have to choose
for each $j<i$ a larger integer and that's all since the other elements have
to be mapped to $i$). Also $|S_n|\leq (n-1)!$ as well. Summing up we get an
upper bound for these semigroups
\[
\sum_{i=1}^n\frac{(n-1)!}{(n-i)!}\ =\ (n-1)!\sum\limits_{j=0}^{n-1}\frac{1}{j!}\ =\ \lfloor e(n-1)!\rfloor,
\]
which comes from the facts that $e=\sum_{j=0}^\infty\frac{1}{j!}$ and $(n-1)!\sum_{j=n}^\infty\frac{1}{j!}< 1$.

For the other case, suppose there exists an $i<j$ and an $i'<j'$ with
$i\neq i'$ such that $if=j$ and $i'f=j'$ do not both hold for any $f\in S_n$.
Still, $i<if$ for each $i<n$ and $nf=n$, by definition of $S_n$ and the assumption
$<=<_n$. The number of such functions satisfying both $if=j$ and $i'f=j'$ is
$\frac{(n-1)!}{(n-i)(n-j)}\geq (n-3)!$, hence the size of $S_n$ is upper-bounded
by $(n-1)!-(n-3)!$. Since $S_n$ is the largest amongst the $S_i$'s and $S$ is the
disjoint union of them we get the claimed upper bound $n((n-1)!-(n-3)!)$.
\end{proof}
We note that the construction for the first case, yielding the upper bound
$\lfloor e(n-1)!\rfloor$ indeed constructs a semigroup $B$ which is exactly the semigroup
from~\cite{brzozo} conjectured there to be a candidate for the maximal-size such subsemigroup.

Our proof can be viewed as a support for this conjecture and can be reformalized as follows:
if there exists some $i$ such that many transformations share this fixed point $i$, then
the size of $S$ is upper-bounded by $\lfloor e(n-1)!\rfloor$ and $S$ is isomorphic to a
subsemigroup of $B$. The question is, whether one can construct a larger semigroup by putting
not too many functions sharing a common fixed point. We also conjecture that $B$ 
is a good candidate for a maximal-size subsemigroup of $T_n$ consisting of nonpermutational
transformations only.

\section{Application to syntactic complexity}
A language $L$ is
\emph{definite} if there exists a constant $k\geq 0$ such that for any $x\in\Sigma^*$, $y\in \Sigma^k$ we have $xy\in L\Leftrightarrow y\in L$ and
is \emph{generalized definite} if there exists a constant $k\geq 0$ such that for any $x_1,x_2\in\Sigma^k$ and $y\in\Sigma^*$ we have $x_1yx_2\in L\Leftrightarrow x_1x_2\in L$.

These are both subclasses of the star-free languages, i.e. can be built from the singletons with repeated use of the concatenation, finite union and complementation operations. It is known that the following decision problem is complete for $\mathbf{PSPACE}$: given a regular language $L$ with its minimal automaton, is $L$ star-free? In contrast, the question for these subclasses above are tractable.

Minimal automata of these languages possess a characterization in terms of \emph{forbidden patterns}.
In our setting, a pattern is an edge-labelled, directed graph $P=(V,E,\ell)$, where $V$ is the set of vertices, $E\subseteq V^2$ is the set of edges,
and $\ell:E\to \mathcal{X}$ is a labelling function which assigns to each edge a variable.
An automaton $\mathbb{A}=(Q,\Sigma,\delta,q_0,F)$ \emph{admits a pattern $P=(V,E,\ell)$}\marginpar{admitting/avoiding a pattern} if there exists an \emph{injective} mapping $f:V\to Q$
and a map $h:\mathcal{X}\to \Sigma^+$ such that for each $(u,v)\in E$ labelled $x$ we have $f(u)\cdot h(x)=f(v)$.
Otherwise $\mathbb{A}$ \emph{avoids} $P$.
  
As an example, consider the pattern $P_d$ on Figure~\ref{fig-patterns}.

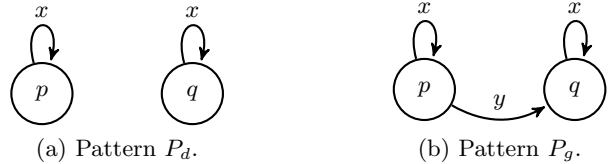
\begin{figure}[h!]
\centering
\subfloat[][Pattern $P_d$.]{
\begin{tikzpicture}[shorten >=1pt,node distance=2cm,>=stealth',thick]
\node[state] (1) {$p$};
\node[state] (2) [right of=1] {$q$};
\draw [->] (1) to[loop above] node[auto] {$x$} (1);
\draw [->] (2) to[loop above] node[auto] {$x$} (2);
\end{tikzpicture}
}
\hfil
\centering
\subfloat[][Pattern $P_g$.]{
\begin{tikzpicture}[shorten >=1pt,node distance=2cm,>=stealth',thick]
\node[state] (1) {$p$};
\node[state] (2) [right of=1] {$q$};
\draw [->] (1) to[loop above] node[auto] {$x$} (1);
\draw [->] (2) to[loop above] node[auto] {$x$} (2);
\draw [->] (1) to[bend right] node[auto] {$y$} (2);
\end{tikzpicture}
}
\caption{Patterns for definite and generalized definite languages.}
\label{fig-patterns}
\end{figure}

A reduced automaton avoids $P_d$ if and only if it recognizes a definite language.
Indeed, a language $L$ is definite iff its syntactic semigroup satisfies the identity
$yx^\omega=x^\omega$. Now assume $L(\mathbb{A})$ admits $P_d$ with $px=p$ and $qx=q$ with $p\neq q$
and $x\in\Sigma^+$. If $q_0x^\omega=p$, then $q_0x^\omega\neq q_0yx^\omega$ for a (nonempty) word
$y$ with $q_0y=q$. If $q_0x^\omega\neq p$, then $q_0x^\omega\neq q_0yx^\omega$ for a (nonempty) $y$
with $q_0y=p$, thus the identity is not satisfied. For the other directon, if
the transition semigroup of an automaton $\mathbb{A}$ does not satisfy $x^\omega=yx^\omega$,
then $p_0x_0^\omega\neq p_0yx_0^\omega$ for some $p_0,x_0$ and $y$; choosing $p=p_0x^\omega$,
$q=p_0y$ and $x_0=x^\omega$
witnesses admittance of $P_d$. (For a more detailed discussion see e.g.~\cite{brzozo}.)

Observe that avoiding $P_d$ is equivalent to state that each nonempty word induces a transformation with at most one fixed point,
which is further equivalent to state that each nonempty word induces a non-permutational transformation:
for each nonempty $u$, the word $u^{|Q|!}$ fixes each state belonging to a nontrivial component of
the graph of $u$, hence $u$ also can have only one state in a nontrivial component, i.e. $u$
induces a nonpermutational transformation. (Again, see~\cite{brzozo} for a different formulation.\footnote{Since -- up to our knowledge -- ~\cite{brzozo} has not been published yet in a peer-reviewed journal or conference proceedings, we include a proof of this fact. Nevertheless, we do not claim this result to be ours, by any means.}.)

Thus Theorem~\ref{thm-main} has the following byproduct:
\begin{corollary}
The syntactic complexity of the definite languages is at most $n((n-1)!-(n-3)!)$.
\end{corollary}

\subsection{The case of generalized definite languages}

In this subsection we show that the syntactic complexity of definite and generalized definite
languages coincide. To this end we study the structure of the minimal automata of the members
of the latter class. In the process we give a (to our knowledge) new (but not too surprising) characterization of the minimal automata of generalized definite languages,
leading to an $\mathbf{NL}$-completeness result of the corresponding decision problem,
as well as a low-degree polynomial deterministic algorithm.

Our first observation is the following characterization:

\begin{theorem}
\label{thm-pattern}
The following are equivalent for a reduced automaton $\mathbb{A}$:
\begin{enumerate}
\item[i)] $\mathbb{A}$ avoids $P_g$.
\item[ii)] Each nontrivial component of $\mathbb{A}$ is a sink, and for each nonempty word $u$ and sink $C$ of $\mathbb{A}$,
  the transformation $u|_C:C\to C$ is non-permutational.
\item[iii)] $\mathbb{A}$ recognizes a generalized definite language.
\end{enumerate}
\end{theorem}
\begin{proof}
Let $\mathbb{A}=(Q,\Sigma,\delta,q_0,F)$ be a reduced automaton.

{\bf i)$\to$ii).} Suppose $\mathbb{A}$ avoids $P_g$.
Suppose that $u|_C$ is permutational for some sink $C$ and word $u\in\Sigma^+$.
Then there exists a set $D\subseteq C$ with $|D|>1$ such that $u$ induces a permutation
on $D$. Then, $x=u^{|D|!}$ is the identity on $D$. Choosing arbitrary distinct states $p,q\in D$
and a word $y$ with $py=q$ (such $y$ exists since $p$ and $q$ are in the same component of $\mathbb{A}$),
we get that $\mathbb{A}$ admits $P_g$ by the $(p,q,x,y)$ defined above, a contradiction.
Hence, $u|_C$ is non-permutational for each sink $C$ and word $u\in\Sigma^+$.

Now assume there exists a nontrivial component $C$ which is not a sink.
Then, $pu=p$ for some $p\in C$ and word $u\in\Sigma^+$. Since $C$ is not a sink, there exists
a sink $C'\neq C$ reachable from $p$ (i.e. all of its members are reachable from $p$).
Since $u$ induces a non-permutational transformation on $C'$, $x=u^{|C'|}$ induces a constant function
on $C'$. Let $q$ be the unique state in the image of $x|_{C'}$. Since $C'$ is reachable from $p$,
there exists some nonempty word $y$ such that $py=q$. Hence, $px=p$, $qx=q$, $py=q$ and $\mathbb{A}$
admits $P_g$, a contradiction.

{\bf ii)$\to$iii).} Suppose the condition of ii) holds.
We show that $L=L(\mathbb{A})$ is generalized definite. 
By the assumption, $q_0u$ belongs to a sink for any $u$ with $|u|\geq |Q|$.
On the other side, viewing a sink $C$ as a (reduced) automaton $\mathbb{C}=(C,\Sigma,\delta|_C,p,F\cap C)$
with $p$ being an arbitrary state of $C$ we get that the transition semigroup of $\mathbb{C}$
consists of nonpermutational transformations only, i.e. $L(\mathbb{C})$ is $k$-definite for some
$k=k_C$. Hence choosing $n$ to be the maximum of $|Q|$ and the values $k_C$ with $C$ being a sink
we get that $L$ is $n$-generalized definite (since the length-$n$ prefix of $u$ determines the sink $C$
to which $q_0u$ belongs and the length-$n$ suffix of $u$, once we know $C$, determines the unique
state in $Cu$).

{\bf iii)$\to$i).} Suppose $L(\mathbb{A})$ is generalized definite.
Then its syntactic semigroup satisfies $x^\omega yx^\omega=x^\omega$ (see e.g.~\cite{handbook}).

Now assume $\mathbb{A}_L$ admits $P_g$ with $px=p$, $qx=q$ and $py=q$ for
the nonempty words $x,y$ and different states $p,q$. Then $p x^\omega = p$ and $px^\omega yx^\omega=q$,
and the identity is not satisfied, thus $L$ is not generalized definite.
\end{proof}

\subsubsection{Complexity issues}

We now take a slight excursion.

Using the characterization given in Theorem~\ref{thm-pattern}, we study the complexity of the following decision problem $\textsc{GenDef}$: given a finite automaton $\mathbb{A}$, is $L(\mathbb{A})$ a generalized definite language?

\begin{theorem}
Problem $\textsc{GenDef}$ is $\mathbf{NL}$-complete.
\end{theorem}
\begin{proof}
First we show that $\textsc{GenDef}$ belongs to $\mathbf{NL}$.
By~\cite{cho-huynh}, minimizing a DFA can be done in nondeterministic logspace.
Thus we can assume that the input is already minimized, since the class of
(nondeterministic) logspace computable functions is closed under composition.

Consider the following algorithm:
\begin{enumerate}
\item Guess two different states $p$ and $q$.
\item Let $s:=p$.
\item Guess a letter $a\in \Sigma$. Let $s:=sa$.
\item If $s=q$, proceed to Step 5. Otherwise go back to Step 3.
\item Let $p':=p$ and $q':=q$.
\item Guess a letter $a\in\Sigma$. Let $p':=p'a$ and $q'=q'a$.
\item If $p=p'$ and $q=q'$, accept the input. Otherwise go back to Step 6.
\end{enumerate}
The above algorithm checks whether $\mathbb{A}$ admits $P_g$: first it guesses $p\neq q$,
then in Steps 2--4 it checks whether $q$ is accessible from $p$,
and if so, then in Steps 5--7 it checks whether there exists a word $x\in\Sigma^+$ with $px=p$ and $qx=q$.
Thus it decides\footnote{Note that in this form, the algorithm can enter an infinite loop which fits
into the definition of nondeterministic log\emph{space}.
Introducing a counter and allowing at most $n$ steps in the first cycle and at most $n^2$ in the second
we get a nondeterministic algorithm using logspace and polytime, as usual.} the complement of $\textsc{GenDef}$, in nondeterministic logspace; since $\mathbf{NL}=\mathrm{co}\mathbf{NL}$,
we get that $\textsc{GenDef}\in\mathbf{NL}$ as well.

For $\mathbf{NL}$-completeness we recall from~\cite{jones-lien-laaser} that the reachability problem for DAGs ($\textsc{DAG-Reach}$)
is complete for $\mathbf{NL}$:
given a directed acyclic graph $G=(V,E)$ on $V=[n]$ with $(i,j)\in E$ only if $i<j$,
is $n$ accessible from $1$?
We give a logspace reduction from $\textsc{DAG-Reach}$ to $\textsc{GenDef}$ as follows.
Let $G=([n],E)$ be an instance of $\textsc{DAG-Reach}$.
For a vertex $i\in[n]$, let $N(i)=\{j:(i,j)\in E\}$ stand for the set of its neighbours and
let $d(i)=|N(i)|<n$ denote the outdegree of $i$. When $j\in[d(i)]$, then the $j$th neighbour of $i$, denoted $n(i,j)$
is simply the $j$th element of $N(i)$ (with respect to the usual ordering of integers of course).
Note that for any $i\in [n]$ and $j\in[d(i)]$ both $d(i)$ and the $n(i,j)$ (if exists) can be computed in logspace.

We define the automaton $\mathbb{A}=([n+1],[n],\delta,1,\{n+1\})$
where 
\[\delta(i,j)=\left\{\begin{array}{ll}
n+1&\hbox{if }(i=n+1)\hbox{ or }(j=n)\hbox{ or }(i<n\hbox{ and }d(i)<j);\\
1&\hbox{if }i=n\hbox{ and }j<n;\\
n(i,j)&\hbox{otherwise.}\\
\end{array}\right.\]
Note that $\mathbb{A}$ is indeed an automaton, i.e. $\delta(i,j)$ is well-defined for each $i,j$.

We claim that $\mathbb{A}$ admits $P_g$ if and only if $n$ is reachable from $1$ in $G$.
Observe that the underlying graph of $\mathbb{A}$ is $G$, with a new edge $(n,1)$ and with a new vertex $n+1$, which is a neighbour of each vertex.
Hence, $\{n+1\}$ is a sink of $\mathbb{A}$ which is reachable from all other states.
Thus $\mathbb{A}$ admits $P_g$ if and only if there exists a nontrivial component of $\mathbb{A}$ which is different from $\{n+1\}$.
Since in $G$ there are no cycles, such component exists if and only if the addition of the edge $(n,1)$ introduces a cycle,
which happens exactly in the case when $n$ is reachable from $1$. Note that it is exactly the case when $1x=1$ for some word $x\in\Sigma^+$.

What remains is to show that the \emph{reduced} form $\mathbb{B}$ of $\mathbb{A}$ admits $P_g$ if and only if $\mathbb{A}$ does.
First, both $1$ and $n+1$ are in the connected part $\mathbb{A}'$ of $\mathbb{A}$, and are distinguishable by the empty word
(since $n+1$ is final and $1$ is not).
Thus, if $\mathbb{A}$ admits $P_g$ with $1x=1$ and $(n+1)x=n+1$ for some $x\in\Sigma^+$, then $\mathbb{B}$ admits $P_g$ with
$h(1)x=h(1)$ and $h(n+1)x=h(n+1)$ (with $h$ being the homomorphism from the connected part of $\mathbb{A}$ onto its reduced form).
For the other direction, assume $h(p)x_0=h(p)$ for some state $p\neq n+1$
(note that since $n+1$ is the only final state, $p\neq n+1$ if and only if $h(p)\neq h(n+1)$).
Let us define the sequence $p_0,p_1,\ldots$ of states of $\mathbb{A}$ as $p_0=p$, $p_{t+1}=p_tx_0$.
Then, for each $i\geq 0$, $h(p_i)=h(p)$, thus $p_i\in[n]$. Thus, there exist indices $0\leq i<j$ with $p_i=p_j$, yielding
$p_ix_0^{j-i}=p_i$, thus $\mathbb{A}$ admits $P_g$ with $p=p_i$, $q=n+1$, $x=x_0^{j-i}$ and $y=n$.

Hence, the above construction is indeed a logspace reduction
from $\textsc{DAG-Reach}$ to the complement of $\textsc{GenDef}$, showing $\mathbf{NL}$-hardness of the latter;
applying $\mathbf{NL}=\mathrm{co}\mathbf{NL}$ again, we get $\mathbf{NL}$-hardness of $\textsc{GenDef}$ itself.
\end{proof}

It is worth observing that the same construction also shows $\mathbf{NL}$-hardness (thus completeness) of the
problem whether the input automaton accepts a definite language.

Thus, the complexity of the problem is characterized from the theoretic point of view.
However, nondeterministic algorithms are not that useful in practice. Since $\mathbf{NL}\subseteq\mathbf{P}$, the
problem is solvable in polynomial time -- now we give an efficient (quadratic) deterministic decision algorithm:

\begin{enumerate}
\item Compute $\mathbb{A}'=(Q,\Sigma,\delta,q_0,F)$, the reduced form of the input automaton $\mathbb{A}$.
\item Compute $\Gamma(\mathbb{A}')$, the component graph of $\mathbb{A}'$.
\item If there exists a nontrivial, non-sink component, reject the input.
\item Compute $\mathbb{B}=\mathbb{A}'\times\mathbb{A}'$ and $\Gamma(\mathbb{B})$.
\item Check whether there exist a state $(p,q)$ of $\mathbb{B}$ in a nontrivial component (of $\mathbb{B}$)
  for some $p\neq q$ with $p$ being in the same sink as $q$ in $\mathbb{A}$. If so, reject the input; otherwise accept it.
\end{enumerate}

The correctness of the algorithm is straightforward by Theorem~\ref{thm-pattern}: after minimization
(which takes $\mathcal{O}(n\log n)$ time) one computes the component graph of the reduced automaton
(taking linear time) and checks whether there exists a nontrivial component which is not a sink
(taking linear time again, since we already have the component graph). If so, then the answer is $\texttt{NO}$.
Otherwise one has to check whether there is a (sink) component $C$ and a word $x\in\Sigma^+$ such that
$f_x|_C$ has at least two different fixed points. Now it is equivalent to ask
whether there is a state $(p,q)$ in $\mathbb{A}'\times\mathbb{A}'$ with $p$ and $q$ being in the same
component and a word $x\in\Sigma^+$ with $(p,q)x=(p,q)$. This is further equivalent to ask whether
there is a $(p,q)$ with $p,q$ being in the same sink such that $(p,q)$ is in a nontrivial component of $\mathbb{B}$.
Computing $\mathbb{B}$ and its components takes $\mathcal{O}(n^2)$ time, and (since we still have the component graph of $\mathbb{A}$)
checking this condition takes constant time for each state $(p,q)$ of $\mathbb{B}$, the algorithm consumes a total of $\mathcal{O}(n^2)$
time.

Hence we have an upper bound concluding this excursion:
\begin{theorem}
Problem $\textsc{GenDef}$ can be solved in $\mathcal{O}(n^2)$ deterministic time in the RAM model of computation.
\end{theorem}

\subsubsection{Syntactic complexity}
In~\cite{brzozo} it has been shown that the class of definite languages has syntactic complexity $\geq \lfloor e\cdot(n-1)!\rfloor$,
thus the same lower bound also applies for the larger class of generalized definite languages.

\begin{theorem}
The syntactic complexity of the definite and that of the generalized definite languages coincide.
\end{theorem}
\begin{proof}
It suffices to construct for an arbitrary reduced automaton $\mathbb{A}=(Q,\Sigma,\delta,q_0,F)$ recognizing a generalized definite language
a reduced automaton $\mathbb{B}=(Q,\Delta,\delta',q_0,F')$ for some $\Delta$
recognizing a definite language such that $|\mathcal{T}(\mathbb{A})|\leq |\mathcal{T}(\mathbb{B})|$.

By Theorem~\ref{thm-pattern}, if $L(\mathbb{A})$ is generalized definite and $\mathbb{A}$ is reduced, then $Q$ can be partitioned as a disjoint union
$Q=Q_0\uplus Q_1\uplus\ldots\uplus Q_c$ for some $c>0$ such that each $Q_i$ with $i\in[c]$ is a sink of $\mathbb{A}$ and $Q_0$ is the
(possibly empty) set of those states that belong to a trivial component. Without loss of generality we can assume that
$Q=[n]$ and $Q_0=[k]$ for some $n$ and $k$, and that for each $i\in[k]$ and $a\in \Sigma$, $i<ia$.
The latter condition is due to the fact that reachability restricted to the set $Q_0$ of states in trivial components is a
partial ordering of $Q_0$ which can be extended to a linear ordering.
Clearly, if $Q_0$ is nonempty, then by connectedness $q_0=1$ has to hold; otherwise $c=1$ and we again may assume $q_0=1$.
Also, $Q_i\Sigma\subseteq Q_i$ for each $i\in[c]$, and let $|Q_1|\leq |Q_2|\leq\ldots\leq|Q_c|$.

Then, each transformation $f:Q\to Q$ can be uniquely written as the source tupling $[f_0,\ldots,f_c]$ of some functions $f_i:Q_i\to Q$
with $f_i:Q_i\to Q_i$ for $0<i\leq c$.
For any $[f_0,\ldots,f_c]\in\mathcal{T}=\mathcal{T}(\mathbb{A})$ the following hold: $f_0(i)>i$ for each $i\in[k]$, and $f_j$ is non-permutational
on $Q_j$ for each $j\in[c]$.
For $k=0,\ldots,c$, let $\mathcal{T}_k$ stand for the set $\{f_k:f\in\mathcal{T}\}$ (i.e. the set of functions $f|_{Q_k}$ with
$f\in\mathcal{T}$). Then, $|\mathcal{T}|\leq \mathop\prod\limits_{0\leq k\leq c}|\mathcal{T}_k|$.

If $|Q_c|=1$, then all the sinks of $\mathbb{A}$ are singleton sets.
Thus there are at most two sinks, since if $C$ and $D$ are singleton sinks whose
members do not differ in their finality, then their members are not distinguishable, thus $C=D$ since $\mathbb{A}$ is reduced.
Such automata recognize reverse definite languages,
having a syntactic semigroup of size at most $(n-1)!$ by \cite{brzozo}, 
thus in that case $\mathbb{B}$ can be chosen to an arbitrary definite automaton having $n$ state and a syntactic semigroup of size
at least $\lfloor e(n-1)!\rfloor$ (by the construction in \cite{brzozo}, such an automaton exists).
Thus we may assume that $|Q_c|>1$. (Note that in that case $Q_c$ contains at least one final and at least one non-final state.)

Let us define the sets $\mathcal{T}'_k$ of functions $Q_i\to Q$ as $\mathcal{T}'_0$ is the set of all elevating functions from $[k]$ to $[n]$,
$\mathcal{T}'_c=\mathcal{T}_c$ and for each $0<k<c$, $\mathcal{T}'_k=Q_c^{Q_k}$. Since $\mathcal{T}_k\subseteq Q_k^{Q_k}$ and $|Q_k|\leq |Q_c|$
for each $k\in[c]$, we have $|\mathcal{T}_k|\leq|\mathcal{T}'_k|$ for each $0\leq k\leq c$. Thus defining $\mathcal{T}'=\{[f_0,\ldots,f_c]:f_i\in\mathcal{T}'_i\}$ it holds that $|\mathcal{T}|\leq|\mathcal{T}'|$.

We define $\mathbb{B}$ as $(Q,\mathcal{T}',\delta',q_0,F)$ with $\delta'(q,f)=f(q)$ for each $f\in\mathcal{T}'$. We show that
$\mathbb{B}$ is a reduced automaton avoiding $P_d$, concluding the proof.

First, observe that $\mathbb{B}$ has exactly one sink, $Q_c$, and all the other states belong to trivial components
(since by each transition, each member of $Q_0$ gets elevated, and each member of $Q_i$ with $0<i<c$ is taken into $Q_c$).
Hence if $\mathbb{B}$ admits $P_d$, then $pt=p$ and $qt=q$ for some distinct pair $p,q\in Q_c$ of states and $t=[t_0',\ldots,t_c']\in\mathcal{T}'$.
This is further equivalent to $pt'_c=p$ and $qt'_c=q$ for some $p\neq q$ in $Q_c$ and $t'_c\in\mathcal{T}'_c$.
By definition of $\mathcal{T}'_c=\mathcal{T}_c$, there exists a transformation of the form $t=[t_0,\ldots,t_{c-1},t'_c]\in\mathcal{T}$
induced by some word $x$, thus $px=p$ and $qx=q$ both hold in $\mathbb{A}$,
and since $p,q$ are in the same sink, there also exists a word $y$ with $py=q$. Hence $\mathbb{A}$ admits $P_g$, a contradiction.

Second, $\mathbb{B}$ is connected. To see this, observe that each state $p\neq 1$ is reachable from $1$ by any transformation of the form
$t=[f_p,t_1,\ldots,t_c]$ where $f_p:[k]\to[n]$ is the elevating function with $1f_p=p$ and $if_p=n$ for each $i>1$.
Of course $1$ is also trivially reachable from itself, thus $\mathbb{B}$ is connected.

Also, whenever $p\neq q$ are different states of $\mathbb{B}$, then they are distinguishable by some word.
To see this, we first show this for $p,q\in Q_c$. Indeed, since $\mathbb{A}$ is reduced, some transformation $t=[t_0,\ldots,t_c]\in\mathcal{T}$
separates $p$ and $q$ (exactly one of $pt=pt_c$ and $qt=qt_c$ belong to $F$). Since $\mathcal{T}_c=\mathcal{T}'_c$, we get that
$p$ and $q$ are also distinguishable by in $\mathbb{B}$ by any transformation of the form $t'=[t_0',\ldots,t_{c-1}',t_c]\in\mathcal{T}'$.
Now suppose neither $p$ nor $q$ belong to $Q_c$. Then, since $\{[t_0',\ldots,t_{c-1}']:t_i'\in\mathcal{T}_i'\}=Q_c^{Q\backslash Q_c}$,
and $|Q_c|>1$, there exists some $t=[t_0',\ldots,t_{c-1}']$ with $pt\neq qt$, thus any transformation of the form
$[t_0',\ldots,t_{c-1}',t_c]\in\mathcal{T}'$ maps $p$ and $q$ to distinct elements of $Q_c$, which are already known to be distinguishable,
thus so are $p$ and $q$. Finally, if $p\in Q_c$ and $q\notin Q_c$, then let $t_c\in\mathcal{T}_c$ be arbitrary and
$t'=[t_0',\ldots,t_{c-1}]\in Q_c^{Q\backslash Q_c}$ with $qt'\neq pt_c$. Then $[t',t_c]$ again maps $p$ and $q$ to distinct states of $Q_c$.

Thus $\mathbb{B}$ is reduced, concluding the proof: $\mathbb{B}$ is a reduced automaton recognizing a definite language
and having a syntactic semigroup $\mathcal{T}'$ with $|\mathcal{T}'|\geq|\mathcal{T}|$.
\end{proof}

\section{Conclusion, further directions}
The forbidden pattern characterization of generalized definite languages we gave is not surprising, based on the identities of the pseudovariety
of (syntactic) semigroups corresponding to this variety of languages.
Still, using this characterization one can derive efficient algorithms for checking whether a given automaton recognizes such a language.
Though we could not compute an exact function for the syntactic complexity, we still managed to show that
these languages are not ``more complex'' than definite languages under this metric. Also, we gave a new upper bound for that.

The exact syntactic complexity of definite languages is still open, as well as for other language classes higher in the dot-depth hierarchy --
e.g. the locally (threshold) testable and the star-free languages.


\begin{thebibliography}{99}
\footnotesize
\bibitem{dotdepth}
  R.~S.~Cohen, J.~Brzozowski.
  Dot-Depth of Star-Free Events.
  Journal of Computer and System Sciences 5(1), 1971, 1--16.
\bibitem{brzozo}
  J.~Brzozowski, D.~Liu.
  Syntactic Complexity of Finite/Cofinite, Definite, and Reverse Definite Languages.
  \url{http://arxiv.org/abs/1203.2873}
\bibitem{cho-huynh}
  S.~Cho, D.~T.~Huynh.
  The parallel complexity of finite-state automata problems.
  Inform. Comput. 97, 1–22, 1992.
\bibitem{ciricimrehsteinby}
  M.~\v{C}iri\v{c}, B.~Imreh, M.~Steinby.
  Subdirectly irreducible definite, reverse definite and generalized definite automata.
  Publ. Electrotechn. Fak. Ser. Mat., 10, 1999, 69--79.
\bibitem{gecsegimreh}
  F.~G\'ecseg, B.~Imreh.
  On isomorphic representations of generalized definite automata.
  Acta Cybernetica 15, 2001, 33--44.
\bibitem{ginzburg}
  A.~Ginzburg.
  About some properties of definite, reverse-definite and related automata.
  IEEE Trans. Electronic Computers EC-15, 1966, 809--810.
\bibitem{holzer-konig}
  M.~Holzer, B.~K\"onig.
  On deterministic finite automata and syntactic monoid size.
  Theoretical Computer Science 327(3), 319--347, 2004.
\bibitem{jones-lien-laaser}
  Neil D. Jones, Y. Edmund Lien and William T. Laaser: New problems complete for nondeterministic log space.
  THEORY OF COMPUTING SYSTEMS Volume 10, Number 1 (1976), 1-17.
\bibitem{klima-polak}
  O.~Kl\'\i ma, L.~Pol\'ak.
  Alternative Automata Characterization of Piecewise Testable Languages.
  Accepted to DLT 2013.
\bibitem{krawetz-lawrence-shallit}
  B.~Krawetz, J.~Lawrence, J.~Shallit.
  State Complexity and the Monoid of Transformations of a Finite Set.
  Proc. of Implementation and Application of Automata, LNCS 3317, 2005, 213--224.
\bibitem{limsc}
  B.~Li.
  Syntactic Complexities of Nine Subclasses of Regular Languages.
  Master's Thesis.
\bibitem{perrin}
  D.~Perrin.
  Sur certains semigroupes syntactiques.
  S\'eminaires de l'IRIA, Logiques et Automates, Paris, 1971, 169--177.
\bibitem{petkovic}
  T.~Petkovi\v{c}, M.~\v{C}iri\v{c}, S.~Bogdanovi\v{c}.
  Decomposition of automata and transition semigroups.
  Acta Cybernetica 13, 1998, 385--403.
\bibitem{handbook}
  J-\'E.~Pin.
  Syntactic semigroups.
  Chapter 10 in Handbook of Formal Languages, Vol. I,
  G. Rozenberg et A. Salomaa (eds.), Springer Verlag, 1997, 679--746.
\bibitem{steinby}
  M.~Steinby.
  On definite automata and related systems.
  Ann. Acad. Sci. Fenn., Ser. A I 444, 1969.
\bibitem{Stern}
  J.~Stern.
  Complexity of some problems from the theory of automata.
  Information and Control 66, 1985, 163--176.
\bibitem{trahtman}
  A.~N.~Trahtman.
  Piecewise and local threshold testability of DFA.
  Proc. of FCT 2001, LNCS 2038 (2001), 347--358.
\end{thebibliography}
\end{document}